\documentclass[submission,copyright,creativecommons]{eptcs}
\providecommand{\event}{EXPRESS/SOS 2013} 
\usepackage{breakurl}             

\usepackage{stmaryrd,latexsym,amssymb,amsmath}
\usepackage{diagrams}
\usepackage{cuijpersstyle,cuijpersdef}

\usepackage{tikz}
\usetikzlibrary{arrows,shapes,snakes,automata,backgrounds,petri,positioning,fit,plotmarks}

\newcommand{\dom}{\mathit{dom}}
\renewcommand{\V}{\ensuremath{\mathbb{V}}}

\newarrow{Iso}C---{>>}
\newarrow{Rel}{.}{.}{.}{.}{.}

\renewcommand{\trans}[4][\empty]{\ensuremath{#2 \, \stackrel{#3}{\rightarrow}_{#1} \, #4}}

\title{The categorical limit \\ of a sequence of dynamical systems}

\author{P.J.L. Cuijpers
\institute{Department of Mathematics and Computer Science\\ Technische Universiteit Eindhoven}
\email{p.j.l.cuijpers@tue.nl}
}

\def\titlerunning{The categorical limit of a sequence of dynamical systems}
\def\authorrunning{P.J.L. Cuijpers}

\begin{document}
\maketitle

\begin{abstract}
Modeling a sequence of design steps, or a sequence of parameter settings, yields a sequence of dynamical systems. In many cases, such a sequence is intended to approximate a certain limit case. However, formally defining that limit turns out to be subject to ambiguity. Depending on the interpretation of the sequence, i.e.\ depending on how the behaviors of the systems in the sequence are related, it may vary what the limit should be. Topologies, and in particular metrics, define limits uniquely, if they exist. Thus they select one interpretation implicitly and leave no room for other interpretations. In this paper, we define limits using category theory, and use the mentioned relations between system behaviors explicitly. This resolves the problem of ambiguity in a more controlled way. We introduce a category of prefix orders on executions and partial history preserving maps between them to describe both discrete and continuous branching time dynamics. We prove that in this category all projective limits exist, and illustrate how ambiguity in the definition of limits is resolved using an example. Moreover, we show how various problems with known topological approaches (in particular that of \cite{Ying}) are now resolved, and how the construction of projective limits enables us to approximate continuous time dynamics as a sequence of discrete time systems.
\end{abstract}

\section{Introduction}

\begin{figure}[ht]
\centering
\begin{tiny}
\begin{diagram}[height=1em,width=1em]
   &&&&&       &&&&&      &            &        &&&&&      &            &      &            &      &&&&&&      &            &      &            &      && \circ   &  &  \\
   &&&&&       &&&&&      &            &        &&&&&      &            &      &            &      &&&&&&      &            &      &            &      &&    &          & \\
   &&&&&       &&&&&      &            &        &&&&&      &            &     &            &       &&&&&&      &            &      &            &     && \uTo   &  &  \\
   &&&&&       &&&&&      &            &        &&&&&      &           &      &           &        &&&&&&      &            &      &            &   &&    && \\
   &&&&&       &&&&&      &            &        &&&&&      &            &      &            & \circ   &&&&&&      &            &   &            & \circ     && \circ  &  &  \\
   &&&&&       &&&&&      &            &        &&&&&      &            &      &            &      &&&&&&      &            &      &            &      &&   &     & \\
   &&&&&       &&&&&      &            &        &&&&&      &           &       &            & \uTo &&&&&&     &            &       &          & \uTo && \uTo   &  &  \\
   &&&&&       &&&&&      &            &        &&&&&      &            &      &            &      &&&&&&      &            &      &            &      &&    &      & \\
   &&&&&       &&&&&      &            &  \circ &&&&&      &            &  \circ  &         &  \circ   &&&&&&      &            & \circ  &      &   \circ   &&  \circ  &  &  \\
   &&&&&       &&&&&      &            &        &&&&&      &            &      &            &      &&&&&&      &            &      &            &      &&    &       & \\
   &&&&&       &&&&&      &            &  \uTo  &&&&&      &            & \uTo &            & \uTo &&&&&&      &            &  \uTo &            & \uTo  && \uTo   &  &  \\
   &&&&&       &&&&&      &            &        &&&&&      &            &      &            &      &&&&&&      &            &      &            &        &&    &        & \\
   &&&&&   \circ  &&&&&  \circ  &             &  \circ  &&&&&   \circ &            &  \circ   &            & \circ   &&&&&&  \circ  &            &  \circ  &            & \circ     && \circ &  &  \\
   &&&&&  \uTo &&&&&      & \luTo(2,3)  & \uTo &&&&&      & \luTo(2,3) & \uTo & \ruTo(2,3) &      &&&&&&      &  \luTo(2,3) & \uTo &  \ruTo(2,3) &      & \ruTo(4,3) &&& \\ \\
\circ &&&&&    \circ &&& &&      &            &  \circ    &&& &&      &            & \circ   &            &      &&&&&&      &            & \circ   &            &      &&    &  &  \\
 \mathit{day}(0) &&&&&  \mathit{day}(1) &&& &&      &            &   \mathit{day}(2)    &&& &&      &            &  \mathit{day}(3)   &            &      &&&&&&      &            &  \mathit{day}(4)   &            &      &&    & \cdots &  \\
\end{diagram}
\end{tiny}
\begin{small}
\caption{A sequence of dynamical systems for which multiple limits may be derived}%
\label{fig:FanSequence1}%
\end{small}
\vspace{0.2cm}
\begin{minipage}[b]{0.45\linewidth}
\centering
\begin{tiny}
\begin{diagram}[height=1.5em,width=2.1em]
      &       &       &       &       & \circ & \cdots \\
      &       &       &       &       & \uTo  & \\
      &       &       &       & \circ & \circ & \cdots \\
      &       &       &       & \uTo  & \uTo  & \\
      &       &       & \circ & \circ & \circ & \cdots \\
      &       &       & \uTo  & \uTo  & \uTo  & \\
      &       & \circ & \circ & \circ & \circ & \cdots \\
      &       & \uTo  & \uTo  & \uTo  & \uTo  & \\
      & \circ & \circ & \circ & \circ & \circ & \cdots \\
      & \uTo  & \uTo  & \uTo  & \uTo  & \uTo  & \\
\circ & \circ & \circ & \circ & \circ & \circ & \cdots \\
& \luTo(3,2) & \luTo(2,2) & \luTo(1,2) \uTo \ruTo(1,2) & \ruTo(2,2) & \ruTo(3,2) \\
      &       &         &          \circ             &            &            & \\
      &       &         &     \mathit{day}(\infty)   &            &            &
\end{diagram}
\end{tiny}
\begin{small}
\caption{Limit obtained by adding a new strand of execution every day.}%
\label{fig:Limit1}%
\end{small}
\end{minipage}
\hspace{0.5cm}
\begin{minipage}[b]{0.45\linewidth}
\centering
\begin{tiny}
\begin{diagram}[height=1.5em,width=2em]
 & \cdots & \cdots & \cdots & \cdots & \cdots & \cdots & \cdots \\
 & \uTo & \uTo & \uTo & \uTo & \uTo & \uTo & \uTo \\
\cdots & \circ & \circ & \circ & \circ & \circ & \circ & \circ \\
 & \uTo & \uTo & \uTo & \uTo & \uTo & \uTo & \uTo \\
\cdots & \circ & \circ & \circ & \circ & \circ & \circ & \circ \\
 & \uTo & \uTo & \uTo & \uTo & \uTo & \uTo & \uTo \\
\cdots & \circ & \circ & \circ & \circ & \circ & \circ & \circ \\
 & \uTo & \uTo & \uTo & \uTo & \uTo & \uTo & \uTo \\
\cdots & \circ & \circ & \circ & \circ & \circ & \circ & \circ \\
        & \luTo(4,2) & \luTo(3,2) & \luTo(2,2) & \luTo(1,2) \uTo \ruTo(1,2) & \ruTo(2,2) & \ruTo(3,2) & \ruTo(4,2) & \\
        &           &         &         &          \circ          &         &   &    & \\
                                 &&&& \mathit{day}(\infty) &&&&
\end{diagram}
\end{tiny}
\begin{small}
\caption{Limit obtained by increasing each strand by one everyday, and adding a new strand of length one.}%
\label{fig:Limit2}%
\end{small}
\end{minipage}
\vspace{-0.2cm}
\end{figure}

As quantitative properties of software systems become increasingly more important, the concept of `approximate correctness' also gains interest. Timed systems should be robust against small deviations in clock speed, hybrid systems should be robust against minor changes in their physical environment, and schedules should be robust against changes in workload and availability of resources. Mathematically, robustness is achieved when approximately equivalent systems have the same properties as the target system, and approximation is usually defined using a metric, or other topology on system behavior (see e.g. \cite{vanBreugel,vanBreugel05,Henzinger10,Huth05,PerrinPin,Ying}).

With a metric, or other topology comes a notion of limit, in the sense that a sequence of systems is said to approximate a certain limit point if the elements of the sequence get closer to that point as the sequence progresses \cite{Eisenberg}. Reversely, with an appropriately chosen notion of limit (a so-called Moore-Smith limit) also comes a topology, in the sense that the closed sets of a topology are exactly those sets that are closed under taking limits \cite{Kelley}.

A recurring problem when defining approximation is that we invariably encounter examples of sequences that we expect to have a certain limit, but within the theory turn out to have a different limit, or no limit at all. One reason for this is that some properties of the sequence would not be preserved in the expected limit (a reasonable argument). Another reason is that the model of execution is not rich enough to express a certain limit (a flaw that we hope to overcome).

An even bigger problem, is that in some cases we can think of multiple reasonable limits for a given sequence, while a metric or (Hausdorff) topology allows us to pick only one! In figure \ref{fig:FanSequence1}, we show an example in which one sequence obtains two different limits, depending on how each system evolves into the next. If we interpret the sequence as gaining an additional strand of length $n$ on every $n$'th day, we obtain as a limit a system with strands of any possible finite length, depicted in figure \ref{fig:Limit1}. If we interpret the sequence as gaining an additional step on each strand every day, and an additional new strand of length $1$, this leads to infinitely many strands of infinite length, depicted in figure \ref{fig:Limit2}. In fact, many other interpretations are also conceivable, leading to even more possible limits.

One way to resolve this, is to assume a deterministic labeling of all executions of systems in the sequence, and take this labeling into account when defining limits. This is the approach taken in \cite{Ying}, where a general notion of limit turns out to have the Moore-Smith properties only if we restrict ourselves to deterministic systems. However, this approach leaves little room for abstraction. In \cite{Ying} it is proposed to use the notion of limit for non-deterministic processes as well, even though it then does not result in a topology. That proposal has certain disadvantages, discussed in more detail in the next section.

Another way to resolve the problem, is to use a refinement order on systems (see e.g. \cite{Morgan09}). If we can find a refinement order that is a \emph{complete partial order (cpo)}, we can define the limits of a (rising) sequence of systems to be any supremum of that sequence. In this way, multiple limits can in theory be defined, but these limits can not be related to each other. In our example, it is unclear what the type of relation should be, given that the elements in the sequence must be related and must all be smaller than the proposed limits, but the limits themselves should not be related. Clearly, the system in figure \ref{fig:Limit1} is in many ways 'smaller' than that in figure \ref{fig:Limit2}. At first sight, brewing up a refinement order will not help us here. But admittedly, we have not searched much further in this direction after we found out that the categorical approach was a promising solution.

In this paper, we recall how the categorical definition of projective limits allows us to take differences in interpretation into account by looking at the morphisms, i.e. the structure preserving maps, between the systems in the sequence. Using projective limits we avoid the implicit choice for one particular interpretation forced by using a metric or topology. In the category we develop, all projective limits turn out to exist, but we leave it as future work to determine which properties carry over to these limits, as this depends also on the further properties of the morphisms that are used in the sequence.

As a category, we consider executions of systems under their natural prefix order (see also \cite{Cuijpers13a}), and find that partial history preserving maps can be used to represent the relation between compositions of behavioral systems and their components. The projective limit of a sequence of such partial maps then gives us the result of a sequence of refinements by composition. We contribute to the theory of dynamical, computational, and hybrid systems, by using a model of behavior that is independent of the model of time one has in mind. As a result, our notion of projective limit not only resolves differences in interpretation, but also allows for continuous time behavior to be represented as the limit of a sequence of discrete time refinements. Using projective limits, we can turn a sequence of computation trees into a continuous system that is not a tree anymore.

While reading this paper, a certain familiarity with category theory \cite{LawvereSchanuel,MacLane} is necessary to understand all technicalities. However, we have attempted to present the material in such a way that the general intuition can be understood with only a basic awareness of category theory. Also, the reader should note that this is not an abstract category theory paper, in that we do not strive for new category theoretic insights (like for example \cite{AbramskyGayNagarajan96}).

Our work certainly relates to the category theoretic work of \cite{CorradiniMontanariRossi96} who consider morphisms between graphs, and \cite{WinskelNielsen97,Cattani99} who use the notion of a presheave to capture histories over a category. But the graph approach considers a path as a composition of steps from one vertex, or state, to another, and in \cite{CuijpersReniers08} we showed that continuous behavior can only be represented in a step-based fashion if it is \emph{finite set refutable} (see also section \ref{sec:continuous}). Therefore, graphs are not expressive enough for our purposes. In contrast, presheaves contain much more information about a system than prefix orders do. Prefix orders only remember the behavior of a single dynamical system upto a certain point in time, and morphisms into other systems are necessary to convey the idea of composition. In a presheave, an object maps to all possible histories that lead to this object, and multiple objects may map the the same set of histories to indicate that those states are executed in parallel. In this paper, we take the view that the semantics of a system consists only of the dynamics. Therefore, the presheave view is too strong for us.

The further structure of this paper is as follows. First, we discuss the notion of limit introduced in \cite{Ying} as an example of a topological approach to defining limits and its inherent difficulties. Secondly, we introduce the category of prefix orders and partial history preserving maps as a way to model the relation between behavior of compositions and components of a dynamical system. Thirdly, we characterize projective limits in this category; fourthly, we apply this characterization to the example discussed above; fifthly, to the examples discussed in the second section; and finally, we show how continuous behavior can be obtained as the projective limit of a sequence of discrete behaviors.

\vspace{-0.2cm}

\section{Bisimulation topologically} \label{sec:ying}

In \cite{Ying} a first attempt is made to define topologically what the limit of a sequence, or more generally a net, of processes is. It is also observed in that book that the proposed definition has certain problems. In this section, we will briefly recast the definitions of \cite{Ying} to transition systems, and recall the problems with those definitions. In this paper, we do not delve deeply into other topological approaches to the theory of computation (see e.g. \cite{vanBreugel,vanBreugel05,Henzinger10,Huth05,PerrinPin}), but all those approaches at least share the problem that we would like to be able to specify multiple distinguishable limits of the same sequence.

\begin{definition}[Labeled transition system]
A \emph{labeled transition system} is a tuple $\langle X,A,i,\rightarrow \rangle$, consisting of a set of states $X$, a set of observables $A$, an initial state $i \in A$, and a transition relation $\rightarrow \subseteq X \times A \times X$. Given $a \in A$ we write $\trans{x}{a}{x'}$ for $(x,a,x') \in \rightarrow$. A labeled transition system is \emph{deterministic} if for each $x \in \X$ and $a \in A$ there is at most one $x' \in X$ such that $\trans{x}{a}{x'}$.
\end{definition}

\begin{definition}[Bisimulation]
Two labeled transition systems $\langle X,A,i,\rightarrow \rangle$ and $\langle Y,A,j,\rightarrow \rangle$ are called \emph{bisimilar} if there exists a relation $\mathcal{R} \subseteq X \times Y$ between the sets of states $X$ and $Y$ such that:
\begin{itemize}
\item $i \, \mathcal{R} \, j$;
\item If $x \ \mathcal{R} \ y$ and $\trans{x}{a}{x'}$, then there exists a $y'$ such that $x' \ \mathcal{R} \ y'$ and $\trans{y}{a}{y'}$;
\item If $x \ \mathcal{R} \ y$ and $\trans{y}{a}{y'}$, then there exists a $x'$ such that $x' \ \mathcal{R} \ y'$ and $\trans{x}{a}{x'}$.
\end{itemize}
\end{definition}

\begin{definition}[Directed family]
A \emph{directed set} $\langle \mathbb{I}, \preceq \rangle$ consists of an \emph{index set} $\mathbb{I}$ and a \emph{relation} $\preceq \subseteq \mathbb{I} \times \mathbb{I}$ that is:
\begin{itemize}
\item reflexive: $\forall_{a \in \mathbb{I}} \ a \preceq a$;
\item transitive: $\forall_{a,b \in \mathbb{I}} \ a \preceq b \ \wedge \ b \preceq c \ \Rightarrow \ a \preceq c$;
\item directed: $\forall_{a,b \in \mathbb{I}} \exists_{c \in \mathbb{I}} \ a \preceq c \ \wedge \ b \preceq c$.
\end{itemize}
A subset $\mathbb{C} \subseteq \mathbb{I}$ is \emph{cofinal} if:
\begin{itemize}
\item cofinality: $\forall_{a \in \mathbb{I}} \exists_{n \in \mathbb{C}} \ a \preceq n$;
\end{itemize}
A \emph{net} (or generalized Moore-Smith sequence) is a family $\{x_i \mid i \in \mathbb{I} \}$ indexed by a directed set $\mathbb{I}$. A \emph{subnet} is then the family $\{x_i \mid i \in \mathbb{C} \}$ over a cofinal subset $\mathbb{C} \subseteq \mathbb{I}$.
\end{definition}

\begin{definition}[Limit bisimulation] \label{def:limitbisim}
Given a directed set $\langle \mathbb{K}, \preceq \rangle$ and a net of labeled transition systems $\{ \langle X_k, A, i_k, \rightarrow_k \rangle \mid k \in \mathbb{K} \}$, a transition system $\langle X_\omega,A,i_\omega,\rightarrow_\omega \rangle$ is called \emph{limit bisimilar} to the net if there exists a relation $\mathcal{R}$ between states in $X_\omega$ and nets of the form $\{ x_k \in X_k \mid k \in \mathbb{C} \}$ with $\mathbb{C}$ a cofinal subset of $\mathbb{K}$ such that:
\begin{itemize}
\item $i_\omega \, \mathcal{R} \, \{ i_k \mid k \in \mathbb{K} \}$;
\item If $x \ \mathcal{R} \ \{ x_k \mid k \in \mathbb{C} \}$ and $\trans{x}{a}{x'}$ then there exists a $n_0 \in \mathbb{C}$ and a net $\{ x'_k \mid n_0 \leq k \in \mathbb{C} \}$ such that $x' \ \mathcal{R} \ \{ x'_k \mid k \in \mathbb{C} \}$ and for all $k \geq n_0$ we find $\trans{x_k}{a}{x'_k}$;
\item If $x \ \mathcal{R} \ \{ x_k \mid k \in \mathbb{C} \}$ and $\mathbb{D}$ is a cofinal subset of $\mathbb{C}$ and $\trans{x_l}{a}{x'_l}$ for all $l \in \mathbb{D}$, then there exists a $x'$ and a cofinal subset $\mathbb{B}$ of $\mathbb{D}$ such that $x' \ \mathcal{R} \ \{ x_l \mid l \in \mathbb{B} \}$ and $\trans{x}{a}{x'}$.
\end{itemize}
\end{definition}

\paragraph{Problem 1: limits of constant nets}
An important problem with definition \ref{def:limitbisim} that is already acknowledged in \cite{Ying}, is that a given labeled transition system $\langle X,A,i,\rightarrow \rangle$ is not always limit bisimilar to a constant net $\{ \langle X,A,i,\rightarrow \rangle \mid n \in \mathbb{I} \}$. Roughly speaking, it is shown in \cite{Ying} that a constant net will only be limit bisimilar to its value if the `cardinality of the non-determinism' is strictly smaller than the `cardinality of the net' (please refer to \cite{Ying} for the precise definitions). As an example, if we interpret figure \ref{fig:Limit1} as a non-deterministic labeled transition system (just apply the same label to each arrow), then the non-determinism of this system has cardinality $\omega$, and it is not limit bisimilar to the constant sequence over this transition system indexed by $\N$.

\begin{figure}[ht]
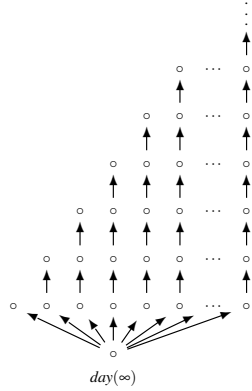

\vspace{-0.4cm}
\centering
\begin{tiny}
\begin{diagram}[height=1.5em,width=2.1em]
      &       &       &       &       &       &        & \vdots \\
      &       &       &       &       &       &        & \uTo \\
      &       &       &       &       & \circ & \cdots & \circ \\
      &       &       &       &       & \uTo  &        & \uTo \\
      &       &       &       & \circ & \circ & \cdots & \circ \\
      &       &       &       & \uTo  & \uTo  &        & \uTo \\
      &       &       & \circ & \circ & \circ & \cdots & \circ \\
      &       &       & \uTo  & \uTo  & \uTo  &        & \uTo \\
      &       & \circ & \circ & \circ & \circ & \cdots & \circ \\
      &       & \uTo  & \uTo  & \uTo  & \uTo  &        & \uTo \\
      & \circ & \circ & \circ & \circ & \circ & \cdots & \circ \\
      & \uTo  & \uTo  & \uTo  & \uTo  & \uTo  &        & \uTo \\
\circ & \circ & \circ & \circ & \circ & \circ & \cdots & \circ\\
& \luTo(3,2) & \luTo(2,2) & \luTo(1,2) \uTo \ruTo(1,2) & \ruTo(2,2) & \ruTo(3,2) & \ruTo(4,2)\\
      &       &         &          \circ             &            &            & \\
      &       &         &     \mathit{day}(\infty)   &            &            &
\end{diagram}
\end{tiny}
\vspace{-0.3cm}
\begin{small}
\caption{Limit bisimulation of figure \ref{fig:FanSequence1} as well as of the constant sequence over figure \ref{fig:Limit1}.}%
\label{fig:Limit3}%
\end{small}
\vspace{-0.2cm}
\end{figure}

An actual limit of this constant sequence is depicted in figure \ref{fig:Limit3}, which is similar to that of figure \ref{fig:Limit1} but with an added infinite strand of transitions. To see that this is indeed a limit, let us represent the initial state with $\bot$, and all other states in figure \ref{fig:Limit1} by pairs $(k,l) \in \N \times \N$ such that $l \leq k$. All transitions are then of the form $\trans{\bot}{}{(k,0)}$ or $\trans{(k,l)}{}{(k,l+1)}$. Furthermore, in the system in figure \ref{fig:Limit3} we have in addition a transition $\trans{\bot}{}{(\omega,0)}$ and transitions $\trans{(\omega,l)}{}{(\omega,l+1)}$, where $\omega$ is the first infinite ordinal. To see that figure \ref{fig:Limit3} is limit bisimilar to the constant sequence over figure \ref{fig:Limit1}, the reader may verify that we indeed get a witnessing limit bisimulation if we take the relation $\mathcal{R}$ such that $x \, \mathcal{R} \,\{x_k \mid k \in \K \}$ if and only if $\K$ is a cofinal subset of $\N$ and
\begin{itemize}
\item $x = \bot$ and $x_k = \bot$ for all $k \in \K$, or
\item $x = (k,l)$ and $\exists_{n \in \N}\forall_{n \leq m \in \K} \ x_m = (k,l)$, or
\item $x = (\omega,l)$ and $\exists_{n \in \N}\forall_{n \leq m \in \K}\exists_{k \in \N} \ x_m = (k,l)$ and $\forall_{k \in \N}\exists_{m \in \K}\exists_{k' \geq k} \ x_m = (k',l)$;
\end{itemize}

Incidentally, using a similar construction we find that the transition system in figure \ref{fig:Limit3} is also limit bisimilar to the sequence in figure \ref{fig:FanSequence1}. Furthermore, in \cite{Ying} it is explicitly noted that the transition systems in figure \ref{fig:Limit1} and \ref{fig:Limit2} are not limit bisimilar to that sequence. Especially for the limit in figure \ref{fig:Limit1} this is acknowledged to be a problem in \cite{Ying}, which we are about to resolve.

\paragraph{Problem 2: introduction of undesired behavior}

A second problem we experience with the notion of limit bisimulation of \cite{Ying}, is that it sometimes introduces new behavior that cannot be explained as resulting from a sequence of implementations. For example, if we take a constant sequence indexed by $\N$ over the the labeled transition system in figure \ref{fig:constant lts}, this sequence turns out to have two limits that are not bisimilar. One limit is, of course, again figure \ref{fig:constant lts} (the non-determinism is of cardinality 2, and as a witnessing relation, we relate every state $x$ to any constant subsequence taking the value $x$). Another limit, however, is the transition system in figure \ref{fig:limit of constant lts}. To see this, the reader may verify that we indeed get a witnessing limit bisimulation if we take the relation $\mathcal{R}$ such that $x \mathcal{R} \{x_k \mid k \in \K \}$ if and only if $\K$ is a cofinal subset of $\N$ and
\begin{itemize}
\item $x \in \{1,2,3,4,5\}$ and $\exists_{n \in \N}\forall_{n \leq m \in \K} \ x_m = x$, or
\item $x = *$ and $\exists_{n \in \N}\forall_{n \leq m \in K} \ x_m \in \{2,3\}$ and $\forall_{n \in \N}\exists_{n \leq m \in \K} \ x_m = 2$ and $\forall_{n \in \N}\exists_{n \leq m \in \K} \ x_m = 3$;
\end{itemize}

Obviously, the labeled transition systems in figures \ref{fig:constant lts} and \ref{fig:limit of constant lts} are not bisimilar. If one would like to create a topological space, this is already bad in itself since it yields a space that is non-Hausdorff (i.e.\ equivalence and topological equivalence of points in the space do not coincide). In this paper, we are not necessarily looking for a Hausdorff topology. But even though we explicitly do want to allow multiple limits to the same sequence, we feel that the introduction of a delayed choice in figure \ref{fig:limit of constant lts} is highly undesirable, since there is no indication in the sequence that this delayed choice should be introduced.

\begin{figure}[ht]
\begin{minipage}[b]{0.45\linewidth}
\centering
\begin{tiny}
\begin{diagram}[height=2em,width=2.1em]
4   &          &   &           & 5 \\
\uTo{b} &          &   &           & \uTo{c} \\
2   &          &   &           & 3 \\
        & \luTo{a} &   & \ruTo{a}  & \\
        &          & 1           &
\end{diagram}
\end{tiny}
\begin{small}
\caption{A constant sequence...}%
\label{fig:constant lts}%
\end{small}
\end{minipage}
\hspace{0.5cm}
\begin{minipage}[b]{0.45\linewidth}
\centering
\begin{tiny}
\begin{diagram}[height=2em,width=2.1em]
4   &          &   &      & 5 \\
\uTo{b} & \luTo{b} &   & \ruTo{c} & \uTo{c} \\
2   &          & * &  & 3 \\
        & \luTo{a} & \uTo{a}  & \ruTo{a}  & \\
        &          & 1    &
\end{diagram}
\end{tiny}
\begin{small}
\caption{... an undesired limit bisimulation.}%
\label{fig:limit of constant lts}%
\end{small}
\end{minipage}
\vspace{-0.5cm}
\end{figure}

\paragraph{Problem 3}
A third problem is that the notion of limit bisimulation of \cite{Ying} does not allow for notions of refinement. If we would like to consider a sequence in which we subsequently split steps into smaller steps, we would like to derive that, in the limit, this gives us a densely ordered set of events. This cannot be blamed on the choice of bisimulation alone, of course, since it also is inherent to the choice of labeled transition systems as a model of computation. \\

To circumvent all three problems, we go from a topological to a category theoretic approach, and switch from labeled transition systems to sets of executions under their natural prefix order. In the next sections, we discuss prefix orders, the category of our choice, in more detail. Subsequently, we study projective limits of nets of dynamical systems and maps between them in this category.

\section{The category of prefix orders and partial history preserving maps}

From this point on, we do not restrict ourselves to labeled transition systems anymore, but consider a generalized model of execution trees as our basic model of dynamics. Notably, we generalize the notion of execution tree to allow for multiple initial states, dense behavior, and even executions that have an infinite history. We do this by just considering the set of executions of a dynamical system, and the natural prefix ordering on those executions (see also \cite{Cuijpers13a}).

\begin{definition}[Prefix order] \label{def:prefix order}
A \emph{prefix order} $\langle \mathbb{U}, \preceq \rangle$ consists of a \emph{set of executions} $\mathbb{U}$ and a \emph{prefix relation} $\preceq \subseteq \mathbb{U} \times \mathbb{U}$ that is
\begin{itemize}
\item reflexive: $\forall_{a \in \mathbb{U}} \ a \preceq a$;
\item transitive: $\forall_{a,b \in \mathbb{U}} \ a \preceq b \ \wedge \ b \preceq c \ \Rightarrow \ a \preceq c$;
\item anti-symmetric: $\forall_{a,b \in \mathbb{U}} \ a \preceq b \ \wedge \ b \preceq a \ \Rightarrow \ a = b$;
\item downward total: $\forall_{a,b,c \in \mathbb{U}} \ (a \preceq c \ \wedge \ b \preceq c) \ \Rightarrow \ (a \preceq b \ \vee \ b \preceq a)$.
\end{itemize}
Two basic operations on executions in a prefix order are the \emph{downward closure} and \emph{upward closure}, determining the \emph{history} and \emph{future} of an execution $u \in \mathbb{U}$ by
\begin{itemize}
\item history: $u^- \triangleq \{ v \in \mathbb{U} \mid v \preceq u \}$;
\item future: $u^+ \triangleq \{ v \in \mathbb{U} \mid u \preceq v \}$;
\end{itemize}
We say a set $U \subseteq \mathbb{U}$ is \emph{prefix closed} if for all $u \in U$ we find $u^- \subseteq U$.
\end{definition}

Note, that if every history $u^-$ is \emph{well-ordered} instead of just totally ordered, we get a \emph{tree} in the sense of \cite{Kunen}. A particular example of such a tree, is of course the case in which every history $u^-$ is finite. In this paper, we are not satisfied considering only well-ordered histories, because we also wish to be able to capture continuous behavior.

To model the relation between composed dynamical systems and their components, we take partial history preserving maps between prefix-orders as morphisms for our category (see also \cite{Cattani99}). The usual composition and identity on partial maps then turns these partial history preserving maps into a category.

\begin{definition}[Partial history preserving maps]
A \emph{morphism} $f : \mathbb{U} \rightarrow \mathbb{V}$ between prefix ordered sets $\U$ and $\V$ is a partial function with a \emph{prefix closed domain} that is \emph{history preserving}, i.e.\ $\forall_{u \in \dom(f)} \ f(u^-) = f(u)^-$, with the obvious lifting $f(A) \triangleq \{ f(a) \mid a \in A \}$ of $f$ to subsets $A \subseteq \mathbb{U}$.
\end{definition}

As an example, consider the executions of a labeled transition system $\langle X,A,i,\rightarrow \rangle$. Observe that the set of strings $A^*$ is prefix ordered by writing $\sigma \preceq \sigma'$ if there exists a string $\sigma''$ such that $\sigma' = \sigma \cdot \sigma''$. Furthermore, we may model an execution, or \emph{run}, of the labeled transition system as a function $\rho : \sigma^- \rightarrow X$ from the history of some string $\sigma \in A^*$ such that $\rho(\epsilon) = i$ and for all $\xi \in A^*$ and $a \in A$ we find $\trans{\rho(\xi)}{a}{\rho(\xi \cdot a)}$. The set of all runs is denoted $\mathrm{Run}(\rightarrow)$, and in turn is prefix ordered by $\rho \preceq \rho'$ iff $\dom(\rho) \subseteq \dom(\rho')$ and $\rho(\xi) = \rho'(\xi)$ for all $\xi \in \dom(\rho)$. Finally, the observation function $\lambda : \mathrm{Run}(\rightarrow) \rightarrow A^*$ simply returns $\lambda(\rho) = \sigma$ whenever $\dom(\rho) = \sigma^-$. Incidentally, this function $\lambda$ is a partial (even total) history preserving map. This suggests that we can consider the alphabet of a transition system as a dynamical system itself, and as a `component' of the labeled transition system: its interface.

In \cite{Cuijpers13a} we have proven that two labeled transition systems $\langle X,A,i,\rightarrow \rangle$ and $\langle Y,A,j,\mapsto \rangle$ are bisimilar if and only if there exists a prefix order $\langle \U, \preceq \rangle$ together with a span (in the style of \cite{JoyalNielsenWinskel93}) of two total and surjective history and future preserving functions $f : \U \rightarrow \mathrm{Run}(\rightarrow)$ and $g : \U \rightarrow \mathrm{Run}(\mapsto)$ such that $\lambda(f(u)) = \lambda(g(u))$ for all $u \in \U$. Such maps model that one dynamical system can be considered a correct implementation of a more abstract specification, and the interpretation of bisimulation is that two specifications are equivalent if they have a common implementation. In this paper, we use a more general type of morphism, because not all relations between compositions and their components can be expected to be surjective, total, or future preserving. Indeed, not even the observation function $\lambda$ is surjective, nor is it future preserving. We only assume that there is a partial history preserving function from compositions to components of which the defined part is prefix closed.

Consider figure \ref{fig:composition}, in which one system (in the middle) is composed of two systems (on either side). The behavior of the two components (left and right) is synchronized until a choice is made to abandon one of the components permanently, while all runs along which a $d$ is observed are blocked. In particular, consider the partial history preserving maps from the composition to each of its components, and notice that they are not total (because some behavior is only related to `the other component'), and not surjective nor future preserving (because some behavior in the component is 'blocked'). However, at any point of execution in the composition it is possible to point out `where you are' in each of the components, unless a choice has been made that this component is no longer needed. Hence the maps are history preserving.

For the time being, we will assume that the semantics of a composition tells us how to construct such a map, somehow. We leave it as a topic of future research to find out how these maps can be created as part of, for example, a structured operational semantics.

\begin{figure}[ht]
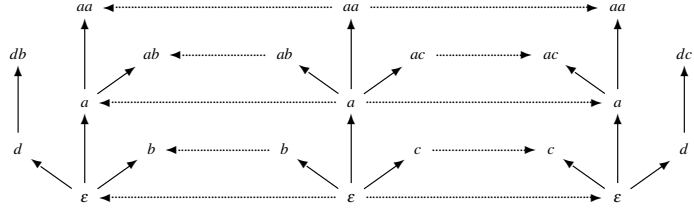

\centering
\vspace{-0.2cm}
\begin{tiny}
\begin{diagram}[height=1.5em,width=2.1em]
      &          & aa       &          &       & & \lDotsto & &       &          & aa       &          &       & & \rDotsto & &       &          & aa       &          &       \\
      &          &          &          &       & &          & &       &          &          &          &       & &          & &       &          &          &          &       \\
db    &          & \uTo     &          & ab    & & \lDotsto & & ab    &          & \uTo     &          & ac    & & \rDotsto & & ac    &          & \uTo     &          & dc    \\
      &          &          & \ruTo    &       & &          & &       & \luTo    &          & \ruTo    &       & &          & &       & \luTo    &          &          &       \\
\uTo  &          & a        &          &       & & \lDotsto & &       &          & a        &          &       & & \rDotsto & &       &          & a        &          & \uTo  \\
      &          &          &          &       & &          & &       &          &          &          &       & &          & &       &          &          &          &       \\
d     &          & \uTo     &          & b     & & \lDotsto & & b     &          & \uTo     &          & c     & & \rDotsto & & c     &          & \uTo     &          & d     \\
      & \luTo    &          & \ruTo    &       & &          & &       & \luTo    &          & \ruTo    &       & &          & &       & \luTo    &          & \ruTo    &       \\
      &          & \epsilon &          &       & & \lDotsto & &       &          & \epsilon &          &       & & \rDotsto & &       &          & \epsilon &          &
\end{diagram}
\end{tiny}
\begin{small}
\caption{Composition by synchronization, choice and blocking, with partial history preserving maps to its components.}%
\label{fig:composition}%
\end{small}
\vspace{-0.3cm}
\end{figure}

In figure \ref{fig:examples}, we give a number of examples of how partial history preserving maps can be used for modeling that one behavioral system is a refinement of another. In section \ref{sec:continuous}, we will use specific refinement maps as an example, when we show how to obtain continuous behavior as the limit of ever more refined discrete behavior.

\begin{figure}[h]
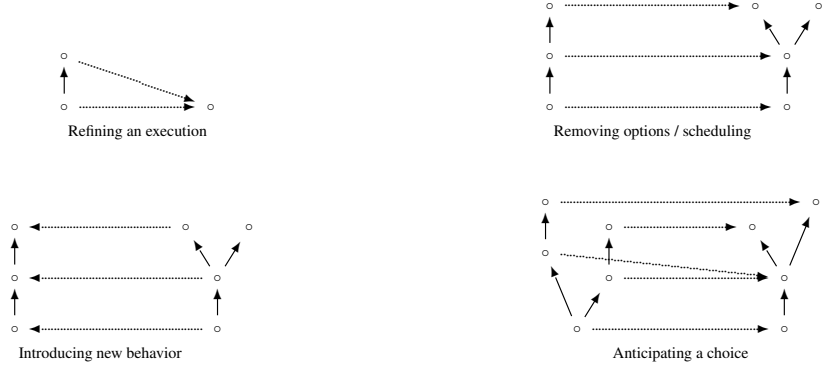

\vspace{-0.5cm}
\centering
\begin{minipage}[b]{0.45\linewidth}
\centering
\begin{tiny}
\begin{diagram}[height=1.6em,width=0.5em]
                                &&&&& \\
                                &&&&& \\
                                &&&&& \\
  \circ &           &             &&& \\
  \uTo  & \rdDotsto &             &&& \\
  \circ & \rDotsto  & \circ       &&& \\
  & \text{Refining an execution} &&&& \\
\end{diagram}
\end{tiny}
\end{minipage}
\begin{minipage}[b]{0.45\linewidth}
\centering
\begin{tiny}
\begin{diagram}[height=1.6em,width=1em]
          &           &       &      &       &       &       \\
    \circ & \rDotsto  & \circ &      &       &       & \circ \\
    \uTo  &           &       &\luTo &       & \ruTo &       \\
    \circ & \rDotsto  &       &      & \circ &       &       \\
    \uTo  &           &       &      & \uTo  &       &       \\
    \circ & \rDotsto  &       &      & \circ &       &       \\
          & \text{Removing options / scheduling} \\
\end{diagram}
\end{tiny}
\end{minipage}
\begin{minipage}[b]{0.45\linewidth}
\centering
\begin{tiny}
\begin{diagram}[height=1.6em,width=1em]
          &           &       &      &       &       &       \\
    \circ & \lDotsto  & \circ &      &       &       & \circ \\
    \uTo  &           &       &\luTo &       & \ruTo &       \\
    \circ & \lDotsto  &       &      & \circ &       &       \\
    \uTo  &           &       &      & \uTo  &       &       \\
    \circ & \lDotsto  &       &      & \circ &       &       \\
          & \text{Introducing new behavior} \\
\end{diagram}
\end{tiny}
\end{minipage}
\begin{minipage}[b]{0.45\linewidth}
\centering
\vspace{-0.2cm}
\begin{tiny}
\begin{diagram}[height=1.6em,width=1em]
  \circ &&       &&       & \rDotsto &       &&       && \circ & \\
  \uTo  &&       && \circ & \rDotsto & \circ &&       &\ruTo(2,3) && \\
  \circ &&       && \uTo  &          &&\luTo &        &        && \\
 & \luTo(2,3) \rdDotsto(8,1) &       && \circ & \rDotsto &       && \circ &&       & \\
 &       &       & \ruTo &&          &       && \uTo  &&       & \\
 &       & \circ &&       & \rDotsto &       && \circ &&       & \\
 &&      &&& \text{Anticipating a choice}         &&&&&& \\
\end{diagram}
\end{tiny}
\end{minipage}
\vspace{-0.3cm}
\begin{small}
\caption{Examples of history preserving maps from refinement to specification}%
\label{fig:examples}
\end{small}
\end{figure}

Finally, the following theorem shows a connection between partial history preserving maps and backward simulation that is useful further on in this paper.
\begin{theorem} \label{thm: backward sim}
A partial function $f : \U \rightarrow \V$ between prefix orders is history preserving and has a prefix-closed domain if and only if it is
\begin{itemize}
\item Order preserving: $\forall_{u,u' \in \dom(f)} \ u \preceq u' \ \Rightarrow \ f(u) \preceq f(u')$;
\item Backward simulation: $\forall_{u \in \dom(f)}\forall_{v \in \V} \ v \preceq f(u) \ \Rightarrow \ \exists_{u' \preceq u} \, f(u') = v$.
\end{itemize}
\end{theorem}
\begin{proof}
Straightforward.
\end{proof}
\vspace{-0.5cm}

\section{The construction of categorical limits}

Now, all preliminaries are in place to get to the definition of projective limit we propose to use to model limits of sequences of compositions.
We first recall the classical category theoretic notion of projective limit over an inverse directed family of morphisms \cite{MacLane}.

\begin{definition}[Inverse directed family]
A family $\{ f_{ij} : \U_j \rightarrow \U_i \mid i,j \in \mathbb{I} \ \wedge \ i \leq j \}$ of morphisms is called an \emph{inverse directed family} if $\mathbb{I}$ is a directed set and
\begin{itemize}
\item reflexive: $\forall_{i \in \mathbb{I}} \ f_{ii}$ is the identity function;
\item transitive: $\forall_{i,j,k \in \mathbb{I}, i \preceq j \preceq k} \ f_{ik} = f_{jk} \circ f_{ij}$.
\end{itemize}
With $\circ$ denoting the usual composition on (partial) functions.
\end{definition}

\begin{definition}[Projective limit]
Given an inverse directed family $\{ f_{ij} : \U_j \rightarrow \U_i \mid i,j \in \mathbb{I} \ \wedge \ i \leq j \}$, its \emph{projective limit} is an object $\V$ together with a family of morphisms $\{ \pi_{i} : \V \rightarrow \U_i \mid i \in \mathbb{I} \}$ (called projections) such that $\forall_{i \leq j} \ \pi_i = f_{ij} \circ \pi_j$, and for any other family $\{ \rho_{i} : \W \rightarrow \U_i \mid i \in \mathbb{I} \}$ of morphisms satisfying $\forall_{i \leq j} \ \rho_i = f_{ij} \circ \rho_j$ there exists a unique morphism $u : \W \rightarrow \V$ such that $\forall_{i \in \mathbb{I}} \, \rho_i = \pi_i \circ u$.
\end{definition}

\noindent Then we characterize this notion for the category of prefix orders and partial history preserving maps.

\begin{definition}[Limit execution]
Given an inverse directed family $\{ f_{ij} : \U_j \rightarrow \U_i \mid i,j \in \mathbb{I} \ \wedge \ i \leq j \}$ of partial history preserving maps, a non-empty set $H \subseteq \biguplus_{i \in \mathbb{I}} \U_i$ (with $\biguplus$ denoting disjoint union of sets) is a \emph{limit execution} if
\begin{itemize}
\item $\forall_{u,v \in H} \forall_{i,j \in \mathbb{I}} \ ( i \leq j \ \wedge \ u \in \U_i \ \wedge \ v \in \U_j ) \ \Rightarrow \ (f_{i,j}(v) = u)$;
\item $\forall_{u \in H} \forall_{i,j \in \mathbb{I}} \ (i \leq j \ \wedge \ u \in \dom(f_{i,j})) \ \Rightarrow \ ( f_{i,j}(u) \in H )$;
\item $\forall_{u \in H} \forall_{i,j \in \mathbb{I}} \ (i \leq j \ \wedge \ u \in \U_i) \ \Rightarrow \ ( \exists_{v \in \U_j} \ v \in H )$.
\end{itemize}
Observe that for every $i \in \mathbb{I}$ the set $H \cap \U_i$ contains at most one element, hence we can define $\pi_i(H)$ as the partial map that returns this element.
Finally, for two limit executions $H$,$K$ we define $H \sqsubseteq K$ if and only if $\pi_i(H) \preceq_i \pi_i(K)$ whenever both $H,K \in \dom(\pi_i)$.
\end{definition}

Note how, if we depict an inverse directed family graphically as in figures \ref{fig:PrefixSequence1} and \ref{fig:PrefixSequence2}, the limit executions become those 'horizontal connecting lines' that stretch into infinity. The first bullet in the definition ensures that, in a limit execution, there is at most one point of execution for each prefix order in the family and that all points are `connected' by the maps in the family. The second bullet ensures that all limit executions 'stretch out' to the left as far as possible, while the third bullet ensures that all limit executions 'stretch out' into infinity.

\begin{theorem}
The set of limit executions of an inverse directed family of partial history preserving maps $\{ f_{ij} : \U_j \rightarrow \U_i \mid i,j \in \mathbb{I} \ \wedge \ i \leq j \}$, denoted by $\underleftarrow{\lim} f_{ij}$, is prefix ordered by $\sqsubseteq$, and the partial maps $\pi_i$ are history preserving. Together they form the projective limit of the family.
\end{theorem}
\begin{proof}
Firstly, we prove that $\sqsubseteq$ indeed gives a prefix order. For this, let $H,H',H'' \in \underleftarrow{\lim} f_{i,j}$.
\begin{itemize}
\item \textbf{Reflexivity:} It trivially follows that $H \sqsubseteq H$.
\item \textbf{Transitivity:} Assume $H \sqsubseteq H'$ and $H' \sqsubseteq H''$. Furthermore, assume that $\pi_i(H)$ and $\pi_(H'')$ are both defined.
For sufficiently large $j \geq i$ we know that $\pi_j(H) \preceq_j \pi_j(H')$ and $\pi_j(H') \preceq_j \pi_j(H'')$, hence by transitivity of $\preceq_j$ we have $\pi_j(H) \preceq_j \pi_j(H'')$. Furthermore, we know that $f_{i,j}$ is history preserving hence order preserving, and that $\pi_i(H) = f_{i,j}(\pi_j(H))$ and $\pi_i(H'') = f_{i,j}(\pi_j(H''))$, so $\pi_i(H) \preceq_i \pi_(H'')$.
\item \textbf{Anti-symmetry:} Assume $H \sqsubseteq H'$ and $H' \sqsubseteq H$, for $H,H' \in \underleftarrow{\lim} f_{i,j}$. Then whenever $\pi_i(H)$ and $\pi_i(H')$ are both defined we find $\pi_i(H) \preceq_i \pi_i(H')$ and $\pi_i(H') \preceq_i \pi_i(H)$, so by anti-symmetry $\pi_i(H) = \pi_i(H')$ for sufficiently large $i$. Finally, observe that if $\pi_i(H) = \pi_i(H')$ for some $i$, then (by the second bullet in the definition of limit execution) this also holds for all smaller $j \leq i$, so $H = H'$.
\item \textbf{Downward totality:} Assume $H \sqsubseteq H''$ and $H' \sqsubseteq H''$. For sufficiently large $i$ we know that $\pi_i(H)$, $\pi_i(H')$ and $\pi_i(H'')$ are all defined, and so $\pi_i(H) \preceq_i \pi_i(H'')$ and $\pi_i(H') \preceq_i \pi_i(H'')$. By downward totality of $\preceq_i$ we then find that $\pi_i(H) = \pi_i(H')$, and by similar reasoning as before $H = H'$.
\end{itemize}
Secondly, we prove that the partial maps $\pi_i$ are history preserving using theorem \ref{thm: backward sim}.
\begin{itemize}
\item \textbf{Order preserving:} Trivial by construction of $\pi_i$;
\item \textbf{Backward simulation:} Assume $\pi_i(H) \preceq v_i$ for some $v_i \in \U_i$. Then for all $j \geq i$ we know that $f_{i,j}$ is history preserving and that $f_{i,j}(\pi_j(H)) = \pi_i(H)$, hence there exists arbitrarily large $j$ and $v_j \in \U_j$ with $v_j \preceq_j \pi_j(H)$ and $f_{k,j}(v_j) = v_k$ for all $i \leq k \leq j$. Using the axiom of choice, we thus construct an $H' \in \underleftarrow{\lim} f_{i,j}$ such that $\pi_i(H') = v_i$ and $H' \sqsubseteq H$.
\end{itemize}
Finally, by construction we have $f_{i,j}(\pi_j(H)) = \pi_i(H)$, for all $i \leq j \in \I$ and $H \in \underleftarrow{\lim} f_{i,j}$, so we only need to verify that $\underleftarrow{\lim} f_{i,j}$ is indeed the smallest candidate for the categorical limit.
\begin{itemize}
\item \textbf{Categorical limit:} Assume we have a family $\{ \rho_{i} : \W \rightarrow \U_i \mid i \in \mathbb{I} \}$ of partial history preserving maps satisfying $\rho_i(w) = f_{ij}(\rho_j(w))$ for all $i \leq j \in \I$ and $w \in \W$. Then we find that $H(w) = \{ \rho_i(w) \mid i \in \I \}$ returns a limit execution, satisfying $\rho_i(w) = \pi_i(H(w))$ by construction. It is easy to verify that $H(w)$ is total (hence has a prefix-closed domain) and history preserving. Furthermore, let $H'(w) : \U \rightarrow \underleftarrow{\lim} f_{ij}$ be any other partial history preserving map such that $\rho_i(w) = \pi_i(H'(w))$ for all $i \in I$, then we find that $\pi_i(H'(w)) = \pi_i(H(w))$ for all $i$, hence $H(w) \sqsubseteq H'(w)$ and $H'(w) \sqsubseteq H(w)$, and by antisymmetry, $H(w) = H'(w)$.
\end{itemize}
\end{proof}

Now that we know how to construct projective limits in the category of prefix orders, we can apply this construction to the examples given thoughout this paper.

\section{One sequence, two limits}

Before, we interpreted the figures \ref{fig:FanSequence1}, \ref{fig:Limit1}, \ref{fig:Limit2} and \ref{fig:Limit3} as labeled transition systems. However, they already have a tree structure, so we can interpret them directly as prefix orders as well. In the introduction we explained how the sequence in figure \ref{fig:FanSequence1} can be interpreted in multiple ways. One way is to say that on each day a new strand is added to the system, leaving the existing executions as they are. In figure \ref{fig:PrefixSequence1} we have modeled this by mapping all executions back to their originals, except for the new strand, which remains unmapped. Another way is to say that on each day each strand is lengthened by one step and a new strand of length $1$ is added. This is modeled by the maps in figure \ref{fig:PrefixSequence2}.

\begin{figure}[h]
\centering
\begin{tiny}
\begin{diagram}[height=0.8em,width=1em]
      &&&&&&          &&      &&&&&&&          &&      &&&&&&&          &&        &&    &&&&&&&&& & &&         &&&          & \cdots\\
      &&&&&&          &&     &&&&&&&&           &      &&&&&&&&          &        &&  &&&&&&&          &      &&    && \circ          &&& \ldDotsto(4,1) &  \\
      &&&&&&          &&     &&&&&&&          &&      &&&&&&&&&  &  &&&&&&  &&&     &&& & \uTo         &&&          &  \\
      &&&&&&          &&      &&&&&&&          &&      &&&&&&&          &&       && &&&&&&& & &&    &&           &&&          & \cdots\\
      &&&&&&          &&      &&&&&&&          &&      &&&&&&&&          &       &&    &&&&&&&          &      && \circ  &&            &&& \ldDotsto(6,1) &  \\
      &&&&&&          &&      &&&&&&&         &&  &&&&&&& & &   &&   &&&&&&&&& \ldDotsto(10,2)& \uTo &&           &&&          & \cdots\\
      &&&&&&          &&     &&&&&&&&           &    &&&&&&&&          &        && \circ  &&&&&&&          &      &&    && \circ          &&& \ldDotsto(4,1) &  \\
      &&&&&&          && &&&&&&&& &  &&&&&&&&& & & \uTo &&&&&&  &&     &&&  & \uTo         &&&          & \\
      &&&&&&          &&      &&&&&&&          &&      &&&&&&&       &&&& &&&&&&      &&      &&    &&           &&&          & \cdots\\
      &&&&&&          &&      &&&&&&&          &&      &&&&&&&         &&        &&    &&&&&&&          &   \circ  &&    &&            &&& \ldDotsto(8,1) &  \\
      &&&&&&          &&      &&&&&&& && &&&&&&&      &&       && &&&&&&&\ldDotsto(10,2)& \uTo &&    &&           &&&          & \cdots \\
      &&&&&&          &&      &&&&&&&        &&      &&&&&&&&          &     \circ  &&    &&&&&&&          &      && \circ  &&            &&& \ldDotsto(6,1) &  \\
      &&&&&&&&&&&&&&  & &&                &&&&&&& &\ldDotsto(9,2)&   \uTo &&  &&&&&&&&& \ldRel
      (10,2) & \uTo &&           &&&          & \cdots\\
      &&&&&&          &&     &&&&&&&&           &   \circ  &&&&&&&&          &        && \circ  &&&&&&&          &      &&    && \circ          &&& \ldDotsto(4,1) &  \\
      &&&&&&& &            &&&&&&&&& \uTo &&&&&&&&&&  & \uTo &&&&&&  &&     &&&  & \uTo        &&&          & \cdots\\
      &&&&&&                 &&      &&&&&&&          &&      &&&&&&&          &&        &&    &&&&&&       \circ &&      &&    &&            &&& \ldDotsto(10,1) &  \\
      &&&&&&                 &&      &&&&&&&          &&      &&&&&&&    &&  &&&&&&&\ldDotsto
      (10,2)&& \luTo(2,8) &      &&    &&           &&&          & \cdots\\
      &&&&&&                 &&      &&&&&&&          &&      &&&&&&&       \circ &&        &&    &&&&&&&          &   \circ &&    &&            &&& \ldDotsto(8,1) &  \\
      &&&&&&                               &&    &&&&&&&&&&&&&&& \ldDotsto(9,2) && \luTo(2,6) &        &&&&&&&&&\ldDotsto(10,2)&      &&    &&           &&&          & \cdots\\
      &&&&&&                &&      &&&&&&&       \circ && &&&&&&&&          &     \circ &&    &&&&&&&          &      && \circ &&            &&& \ldDotsto(6,1) &  \\
      &&&&&&&       &       &&&&&& \ldDotsto(7,2) && \luTo(2,4) &    &&&&&&&& \ldDotsto(9,2) &  &&    &&&&&&    &&& \ruTo(2,4) \ldDotsto(10,2) &&&&&&          & \cdots\\
      &&&&&&&                & \circ &&&&&&&&            &  \circ  &&&&&&&&          &        && \circ &&&&&&&          &      &&    && \circ         &&& \ldDotsto(4,1) &  \\
      &&&&&&&             & \uTo  &&&&&&&&& \uTo &&&&&&&&& \uTo &  \ruTo(2,2) &    &&&&&&  && \uTo &&&  \ruTo(4,2) &            &&&          &  \\
\circ &&&&&&& \lDotsto       & \circ &&&&&&&& \lDotsto &     \circ  &&&&&&&& \lDotsto & \circ     &&    &&&&&&& \lDotsto & \circ   &&    &&            &&& \lDotsto & \cdots \\
\end{diagram}
\end{tiny}
\vspace{-0.2cm}
\begin{small}
\caption{Partial maps obtained by adding a new strand everyday.}%
\label{fig:PrefixSequence1}%
\end{small}
\centering
\vspace{-0.2cm}
\begin{tiny}
\begin{diagram}[height=0.8em,width=1em]
   &&&&&&           &&      &&&&&&&&           &      &&&&&&&&          &        &&   &&&&&&&&&  &      &&    && \circ  & \lDotsto & \cdots \\
   &&&&&&           &&       &&&&&&&          &&      &&&&&&&&& &&  &&&&&&&&  && &&&   & \uTo &          & \\
   &&&&&&           &&       &&&&&&&          &&      &&&&&&&&          &       &&    &&&&&&&&&  &      && \circ  &&    & \lDotsto & \cdots \\
   &&&&&&           &&       &&&&&&&          && &&&&&&&&&    &&    &&&&&&&&    &&&  & \uTo &&    &          & \\
   &&&&&&           &&      &&&&&&&&           &     &&&&&&&&  &        && \circ  &&&&&&&&& \lDotsto &      &&    && \circ  & \lDotsto & \cdots \\
   &&&&&&           &&  &&&&&&&&&  &&&&&&&&&         &  & \uTo &&&&&&&&  &&     &&&            & \uTo &          & \\
   &&&&&&           &&       &&&&&&&          &&         &&&&&&&         &&        &&    &&&&&&&&& &   \circ  &&    &&    & \lDotsto & \cdots \\
   &&&&&&           &&       &&&&&&&     &&&&&&&&&      &&        &&    &&&&&&&& &  & \uTo &&    &&    &          & \\
   &&&&&&           &&       &&&&&&&         &&        &&&&&&&&  &     \circ  &&    &&&&&&&&& \lDotsto &      && \circ  &&    & \lDotsto & \cdots \\
   &&&&&&&&      &&&&&&&      &&         &&&&&&&          &  &   \uTo &&    &&&&&&&&    &&&            & \uTo &&    &          & \\
   &&&&&&           &&    &&&&&&&&   &   \circ    &&&&&&&& \lDotsto &        && \circ  &&&&&&&&& \lDotsto &      &&    && \circ  & \lDotsto & \cdots \\
   &&&&&&           &&       &&&&&&&          && \uTo  &&&&&&&&&         &       & \uTo &&&&&&&&  &&     &&&            & \uTo &   & \\
   &&&&&&           &&       &&&&&&&          &&         &&&&&&&          &&        &&    &&&&&&&&       \circ &&      &&    &&    & \lDotsto & \cdots \\
   &&&&&&           &&       &&&&&&&          &&         &&&&&&&  && &&  &&&&&&&&& \luTo(2,8) &      &&    &&    &          &  \\
   &&&&&&           &&       &&&&&&&          &&         &&&&&&&       \circ &&        &&    &&&&&&&&& \lDotsto &   \circ &&    &&    & \lDotsto  & \cdots \\
   &&&&&&           &&       &&&&&&& && &&&&&&&& \luTo(2,6) &        &&    &&&&&&&&          &&      &&    &&    &          & \\
   &&&&&&           &&       &&&&&&&       \circ &&        &&&&&&&& \lDotsto &     \circ &&    &&&&&&&&& \lDotsto &      && \circ &&    & \lDotsto & \cdots \\
   &&&&&&  &&&&&&&&&& \luTo(2,4) &         &&&&&&&          &&        &&    &&&&&&&&    &&& \ruTo(2,4) &    &&    &          & \\
   &&&&&&           &&   \circ &&&&&&&& \lDotsto  &  \circ    &&&&&&&& \lDotsto &        && \circ &&&&&&&&& \lDotsto &      &&    && \circ & \lDotsto & \cdots \\
&&&&&&& & \uTo &&&&&&&           && \uTo &&&&&&&&& \uTo & \ruTo(2,2) &    &&&&&&&&  && \uTo &&& \ruTo(4,2) &    &          & \\
\circ &&&&&&& \lDotsto &  \circ &&&&&&&& \lDotsto  &  \circ    &&&&&&&& \lDotsto & \circ     &&    &&&&&&&&& \lDotsto & \circ   &&    &&    & \lDotsto & \cdots \\
\end{diagram}
\end{tiny}
\vspace{-0.2cm}
\begin{small}
\caption{Partial maps obtained by adding a step to each strand everyday, and adding one additional strand of length $1$.}%
\label{fig:PrefixSequence2}%
\end{small}
\end{figure}

If we now take the projective limits of the maps in figures \ref{fig:PrefixSequence1} and \ref{fig:PrefixSequence2}, we indeed find the limits of figure \ref{fig:Limit1} and \ref{fig:Limit2}, respectively. Furthermore, we leave it to the reader to try to construct maps such that the limit in figure \ref{fig:Limit3} is obtained.

\begin{theorem} \label{thm: solution example}
The prefix order in figure \ref{fig:Limit1} is the projective limit of figure \ref{fig:PrefixSequence1}, and the prefix order in figure \ref{fig:Limit2} is the projective limit of figure \ref{fig:PrefixSequence2}.
\end{theorem}
\begin{proof}
We formalize the limit in figure \ref{fig:Limit2}, as the set $\mathbb{X} = (\N \times \N) \cup \{ \bot \}$ ordered by $\bot \preceq (x,y)$, for the root, and $(x,y) \preceq (x',y') \ \Leftrightarrow \ x = x' \ \wedge \ y \leq y'$ for all executions $(x,y)$ and $(x',y')$ in $\mathbb{X}$. One may check that this is indeed a prefix order. The limit in figure \ref{fig:Limit1} is in fact a prefix-closed subset of $\ref{fig:Limit2}$, which we model as $\mathbb{Y} = \{ (k,l) \in \mathbb{X} \mid l \leq k \} \cup \{ \bot \}$.

The two different views on figure \ref{fig:FanSequence1} are modeled by constructing two sequences of subobjects from $\X$ that are isomorphic (hence have the same pictorial representation), but lead to different (direct) limits when the natural inclusions are used as implementation relations between them. The first interpretation coincides with the sequence
\[ Y(n) = \{ (k,l) \in \mathbb{X} \mid l \leq k < n \} \cup \{\bot\} \]
while the second interpretation gives us the sequence
\[ X(n) = \{ (k,l) \in \mathbb{X} \mid l \leq n - k - 1 \} \cup \{\bot\}. \]
We find the maps depicted in figures \ref{fig:PrefixSequence1} and \ref{fig:PrefixSequence2}, respectively, by defining $f_{n,m} : Y(m) \rightarrow Y(n)$ as $f_{n,m}(y) = y$ for all $y \in Y(n)$, and undefined elsewhere, and similarly for $g_{n,m} : X(m) \rightarrow X(n)$. It is a standard result from category theory that the partial identity maps constructed this way yield as a projective limit the set of all elements that occur in the sequence from some point onwards. In other words, it is a standard result that $\underleftarrow{\lim} f_{n,m} = \mathbb{Y}$ and $\underleftarrow{\lim} f_{n,m} = \mathbb{X}$, albeit formally we still need to check that the resulting (identity) maps $g_n : X(n) \rightarrow \X$ and $f_n : Y(n) \rightarrow \Y$ are partial history preserving maps. This is easy to verify and left to the reader.
\end{proof}

\section{Avoiding unwanted executions}

The first problem noted in section \ref{sec:ying} is that constant sequences and nets do not always have their value as a limit. It is a standard result from category theory that the projective limit of a inverse directed family of identity morphisms $\mathit{id}_{\U} : \U \rightarrow \U$ is the object $\U$, so we do not have to worry about this in our setting.

The second problem noted in section \ref{sec:ying} is that limits may add unwanted behavior. In itself, it is not necessarily bad if a limit add's new behavior to a sequence of systems, but this behavior must somehow be explained from the behaviors in the sequence. As an example, we take figure \ref{fig:Limit2} as the limit of the sequence in figure \ref{fig:FanSequence1}. Clearly, new -infinite- behavior has been added. But this infinite behavior follows from an explanation (the sequence of partial history preserving maps) that shows how each strand in the system grows as the sequence progresses. So the new infinite behavior is acceptable.

Now, taking the transition system in figure \ref{fig:limit of constant lts} as the limit of a sequence of transition systems in figure \ref{fig:constant lts} strikes us as counter-intuitive, because we do not have a operational explanation of why this delayed choice would follow from the sequence. This may, of course, be blamed on our own lack of imagination. Fact is, however, that if we unfold the behavior of figures \ref{fig:constant lts} and \ref{fig:limit of constant lts} into their prefix orders, depicted in figures \ref{fig:constant lts pfx} and \ref{fig:limit of constant lts pfx}, we find that there does exist an inverse directed family of partial history preserving maps from figure \ref{fig:constant lts pfx} to itself that has figure \ref{fig:limit of constant lts pfx} as its projective limit. More precisely, there does not exist such a family that also preserves the labeling, i.e.\ $\lambda_i(f_{i,j}) = \lambda_j$ for all $i \leq j \in \mathbb{I}$.

\begin{figure}[ht]
\begin{minipage}[b]{0.45\linewidth}
\centering
\begin{tiny}
\begin{diagram}[height=2em,width=2.1em]
ab   &          &          &           & ac \\
\uTo &          &          &           & \uTo \\
a    &          &          &           & a \\
        & \luTo &          & \ruTo  & \\
        &       & \epsilon &
\end{diagram}
\end{tiny}
\begin{small}
\caption{A constant sequence...}%
\label{fig:constant lts pfx}%
\end{small}
\end{minipage}
\hspace{0.5cm}
\begin{minipage}[b]{0.45\linewidth}
\centering
\begin{tiny}
\begin{diagram}[height=2em,width=2.1em]
 ab  & ab\dag &         &   &        & ac\dag &  ac  \\
\uTo &    & \luTo   &   & \ruTo  &    & \uTo \\
  a  &    &         & a\dag &        &    & a    \\
     & \luTo(3,2) & & \uTo & & \ruTo(3,2) &      \\
     &    &    & \epsilon &      &    &      \\
\end{diagram}
\end{tiny}
\begin{small}
\caption{... an impossible projective limit}%
\label{fig:limit of constant lts pfx}%
\end{small}
\end{minipage}
\end{figure}

To see why the system in figure \ref{fig:limit of constant lts pfx} cannot result as a projective limit, we need to study the delayed branching point. I.e.\ we study the middle execution labeled by $a$, and indicated by $a\dag$ in figure \ref{fig:limit of constant lts pfx}, and its two futures $ab\dag$ and $ac\dag$. Assume that $\U$ represents the system in figure \ref{fig:constant lts pfx}, and let $\{ f_{i,j} : \U \rightarrow \U \mid i \leq j \in \I \}$ be any inverse directed family on $\U$ with projective limit $\V$ and projections $\pi_i : \V \rightarrow \U$. For the points $a\dag$,$ab\dag$, and $ac\dag$ to be part of the projective limit, there must exist a $k \in \I$ such that for every $l \geq k$ we have $a\dag,ab\dag,ac\dag \in \dom(\pi_k)$. Furthermore, for $\pi_k$ to be history preserving, $ab\dag$ and $ac\dag$ must be mapped to futures of $\pi_k(a\dag)$ in $\U$, which is obviously not possible if the labeling $\lambda$ is to be preserved (otherwise, we could still map $ab\dag$ and $ac\dag$ both to $ab$, of course).

\section{Continuity as the limit of discrete refinements} \label{sec:continuous}

The third problem addressed in section \ref{sec:ying}, is that limit bisimulations do not provide a way to deal with sequences of refinements. As an example of how this is resolved, we will now give a rather straightforward example of a refinement of discrete behavior that becomes continuous behavior in the limit. Continuous behavior can be modeled by taking flows, i.e.\ continuous functions over the positive real numbers $\R_{\geq 0}$, as executions and prefix order them in the obvious way.

\begin{definition}[Flows]
Given a topological space $\X$, we define $\X \uparrow \R_{\geq 0} = \{ f : t^- \rightarrow \mathbb{X} \mid f \, \text{continuous} \ \wedge \ t \in \R_{\geq 0} \}$ as the set of all \emph{flows} over $\X$, and note that $\X \uparrow \R_{\geq 0}$ is prefix ordered by defining $f \preceq g$ if and only if $\dom(f) \subseteq \dom(g)$ and $f(t) = g(t)$ for any $t \in \dom(f)$ in the obvious way. A continuous system can then be modeled as a prefix-closed set of flows $\Phi \subseteq \ \X \uparrow \R_{\geq 0}$.
\end{definition}

\begin{definition}[Finite set refutability]
A (prefix closed) set of flows $\Phi \subseteq \X \uparrow \R_{\geq 0}$ is \emph{finite set refutable} if for any flow $f : t^- \rightarrow \X$ with $f \not\in \Phi$ there is a finite (i.e. nowhere dense) set $T \subseteq t^-$ such that for every $\phi \in \Phi$ with $\dom{\phi} = \dom{f}$ we find $f(T) \neq \phi(T)$.
\end{definition}

\noindent We can discretize a set of flows using $\epsilon$ steps.

\begin{definition}[Discretization]
Given a set of flows $\Phi \subseteq \X \uparrow \R_{\geq 0}$ and a step-size $0 < \epsilon \in \R_{\geq 0}$, the \emph{$\epsilon$-discretization} of $\Phi$ is the set of sequences $\Phi_\epsilon = \{ f(0) \cdots f(n \cdot \epsilon) \mid f \in \Phi, \, n \in \N, \, n \cdot \epsilon \in \dom(f) \}$, with its natural prefix order on it.
\end{definition}

\noindent Using a sequence of $\frac{1}{2^k}$ discretizations, we then find the following limit theorem.

\begin{theorem} Given a finite-set refutable set of flows $\Phi \subseteq \X \uparrow \R_{\geq 0}$ and a family of refinement maps $\{ f_{k,l} : \Phi_{2^{-l}} \rightarrow \Phi_{2^{-k}} \mid k,l \in \N \ \wedge \ k \leq l \}$ that is recursively defined by
\begin{itemize}
\item $f_{k,k}(\sigma) = \sigma$, for any $\sigma \in \Phi_{2^{-k}}$
\item $f_{k,l+1}(\epsilon) = \epsilon$ and
\item $f_{k,l+1}(x) = \epsilon$, for any sequence $x \in \Phi_{2^{-k}}$ of length $1$, and
\item $f_{k,l+1}(x \cdot x' \cdot \sigma) = f_{k,l}(x \cdot f_{k,l+1}(\sigma))$, for $x,x' \in \X$ and $x \cdot x' \cdot \sigma \in \Phi_{2^{-k}}$
\end{itemize}
we note that this family is history preserving, and find $\Phi = \underleftarrow{\lim} \, f_{k,l}$.
\end{theorem}
\begin{proof}
By construction of the discretization, we find a history preserving map $\phi_k : \Phi \rightarrow \Phi_{2^{-k}}$ from $\Phi$ to every member of the inverse directed family.

It is a property of the categorical limit, that there is a unique map $\phi_\infty : \Phi \rightarrow \underleftarrow{\lim} f_{k,l}$.
This map returns for every $f \in \Phi$ the set of sequences
\[ \phi_\infty(f) = \{ \sigma_i \mid \sigma_i(n) = f(\frac{n}{2^i}) \ \text{whenever} \ \frac{n}{2^i} \in \dom(f) \}.\]
Conversely, we now observe that we can build a map $\phi_\infty^{-1} : \underleftarrow{\lim} f_{k,l} \rightarrow \X \uparrow \R_{\geq 0}$ that for every $H$ returns a function $\phi_\infty^{-1}(H) = h$ with domain $\dom(h) = [0,\sup \{ \frac{n}{2^i} \mid i \in \N \ \wedge \ n \in \dom(\pi_i(H)) \}]$, that is continuous and furthermore satisfies: \[ \forall_{i \in \N} \forall_{n \in \dom(\pi_i(H))} \pi_i(H)(n) = h(\frac{n}{2^i}) \]
To see that this map is a proper inverse, we must verify that every $\phi_{\infty}^{-1}(H) \in \Phi$.
We write $K = \{\frac{n}{2^i} \mid i,n \in \N \}$ and assume that $h = \phi_{\infty}^{-1}(H) \not\in \Phi$. Because $\Phi$ is finite set refutable, we can find a nowhere dense set $T \in \dom(h)$ such that $\phi(T) \neq h(T)$ for any $\phi \in \Phi$. Furthermore, because $h$ is continuous and $K$ is dense in $\R_{\geq 0}$, this means that also $h(K) \neq \phi(K)$. So for every $\phi \in \Phi$ there is a $k \in K$, hence an $i \in \N$, such that $\pi_i(H)(n) = h(\frac{n}{2^i}) \neq \phi(\frac{n}{2^i})$. But picking $H \in \underleftarrow{\lim} f_{k,l}$ assumes that for every $i$ there exists a $\phi$ with $\pi_i(H)(n) = \phi(\frac{n}{2^i})$. A contradiction.

By construction, we find that $\phi_\infty(\phi_\infty^{-1}(H)) = H$. We use the fact that all $f \in \Phi$ are continuous, hence completely defined by their value at time-points of the form $\frac{n}{2^i}$, to find that $\phi_\infty$ is injective and satisfies $\phi_\infty^{-1}(\phi_\infty(H)) = H$. Hence, after verifying that $\phi_\infty$ and $\phi_\infty^-1$ are indeed history preserving, we conclude that $\Phi$ and $\underleftarrow{\lim} \, f_{k,l}$ are isomorphic.
\end{proof}

In \cite{CuijpersReniers08}, we proved that continuous behavior can only be captured as a timed transition system if it is finite set refutable. It was also argued there that all closed physical systems are in fact finite set refutable. An example of a non-finite set refutable system used there applies here as well: consider the differential inclusion $\dot{x} \in \{-1,1\}$. Obviously, the set of solutions of this inclusion contains all `saw-thooth' shaped flows, while the constant functions are examples of functions that are not included but cannot be finitely refuted. If we apply the above construction to this system, we do not find $\{ x \mid \dot{x} \in \{-1,1\} \} = \underleftarrow{\lim} \, f_{k,l}$ as a limit, but instead we get a differential inclusion over an interval: $\{ x \mid \dot{x} \in [-1,1] \} = \underleftarrow{\lim} \, f_{k,l}$. The system has been closed under finite set refutable behavior, which means that all functions that can be arbitrarily approximated by saw-thooths have been added.

We conjecture that it is not possible approximate arbitrary continuous behavior as a limit of a sequence of trees. To approximate arbitrary continuous behavior, we expect to have to revert to nets over prefix orders, and from some point in these nets onwards the prefix orders will not be trees anymore. Still, the result above illustrates how a sequence of trees can turn into a prefix order that is not a tree anymore. This means that, even though we have not researched how to obtain the solutions of $\dot{x} \in \{-1,1\}$ as a limit yet, we can compare this system to the sequence of trees that approximates its closure $\dot{x} \in \{-1,1\}$, without leaving our category. That is why we propose prefix orders as a more suitable semantics for, in particular, hybrid systems, where executions are combine both continuous and discrete behavior.

\section{Concluding remarks and future work}

In this paper, we introduced the category of prefix orders and partial history preserving maps, in which the prefix orders generalize discrete as well as continuous branching time behavior, and the partial history preserving maps capture the relation between the behavior of compositions of dynamical systems and the behavior of their components. Within this category, we studied projective limits, and showed how they may be used to describe the limit of a sequence of dynamical systems. Using examples of problems in the topological theory of bisimulation, some of which already observed in \cite{Ying}, we argued that it is important to take the relation between systems into account while constructing a limit, and we have shown how a category theoretic approach helps to overcome these problems.

The main message of the category theoretic approach is that it is not only important to consider \emph{whether} one system is related to another, but also \emph{how} this relation comes about. In the future, we would like to apply this insight to structured operational semantics as well. One of the issues we left open in this paper, is how to obtain a relation between a composition and its components in a structured way. For example, we hope to be able to extend the operational semantics of process algebraic operations, such that it not only returns a set of executions, but also the implementation relations between processes and their composition. In case of parallel composition and disjoint union, we happen to have a categorical characterization already, as we have shown in \cite{Cuijpers13a} that those operations follow as the categorical product and co-product if we restrict ourselves to total history preserving maps. But for more complex operators like sequential composition, or even history dependent operations (see e.g. \cite{MontanariPistore97,AartsHeidarianVaandrager12}), it would be nice to derive how processes are related to their composition, from their operational semantics.

\textbf{Acknowledgements:}
We would very much like to thank Ruurd Kuiper, Harsh Beohar, Erik de Vink, Reinder Bril and Helle Hansen for their feedback and brainstorming at various stages of this project. Also, we thank Paul Taylor for providing the LaTeX diagram package used to create the figures in this paper.

\bibliographystyle{eptcs}
\bibliography{CuijpersBibliography}

\end{document}